\documentclass[aps,prl,a4paper,showpacs,twocolumn,twoside,superscriptaddress]{revtex4-1}

\usepackage{textcomp,hyperref,enumerate,url,cancel}   
\usepackage[margin=0.8in]{geometry}
\usepackage{color} 
\usepackage{amsmath,amssymb,amsthm,bbm,bm,nicefrac}   				
\usepackage{graphicx,epsfig}                                                               		
\usepackage{multirow}
\usepackage{subfigure}

\newcommand{\ssection}[1]{\emph{#1.---}}

\newcommand{\ket}[1]{\left|#1\right\rangle}							
\newcommand{\bra}[1]{\left\langle#1\right|}

\newcommand{\ketbra}[2]{\left|#1\rangle\!\langle#2\right|}
\newcommand{\braket}[2]{\left\langle #1\lvert#2\right\rangle}

\newcommand{\mean}[1]{\left\langle #1\right\rangle}

\newcommand{\be}{\begin{equation}} 							
\newcommand{\ee}{\end{equation}}
\newcommand{\ba}{\begin{align}}
\newcommand{\ea}{\end{align}}
\newcommand{\bematrix}{\left(\begin{matrix}}
\newcommand{\ematrix}{\end{matrix}\right)}

\theoremstyle{definition}

\theoremstyle{theorem}
\newtheorem{theorem}{Theorem}

\theoremstyle{lemma}

\theoremstyle{proposition}

\theoremstyle{corollary}
\newtheorem{corollary}[theorem]{Corollary}

\theoremstyle{observation}

\theoremstyle{remark}

\newcommand{\1}{\openone}

\def\R{{\ensuremath{\mathbb R}}}

\def\spn{\operatorname{span}}
\def\conv{\operatorname{conv}}
\def\aff{\operatorname{aff}}
\def\diag{\operatorname{diag}}
\newcommand{\tr}{\operatorname{tr}}

\def\cC{\mathcal C}

\def\cF{\mathcal F}

\def\cH{\mathcal H}

\def\cL{\mathcal L}

\def\cP{\mathcal P}
\def\cQ{\mathcal Q}

\newcommand{\aw}[1]{{#1}}

\newcommand{\lle}[1]{{#1}}

\newlength{\blank}
\newenvironment{proofof}[1][{\hspace{-\blank}}]{{\medskip\noindent\textbf{Proof~{#1}.\ }}}{\hfill\qed}

\begin{document}

\title{All tight correlation Bell inequalities have quantum violations}

\author{\aw{Lloren\c{c}} Escol\`a}
\affiliation{F\'isica Te\`orica: Informaci\'o i Fen\`omens Qu\`antics, Departament de F\'isica, Universitat Aut\`onoma de Barcelona, 08193 Bellaterra (Barcelona) Spain}
\email{llorensescola@gmail.com, john.calsamiglia@uab.cat}

\author{\aw{John} Calsamiglia}
\affiliation{F\'isica Te\`orica: Informaci\'o i Fen\`omens Qu\`antics, Departament de F\'isica, Universitat Aut\`onoma de Barcelona, 08193 Bellaterra (Barcelona) Spain}
\email{john.calsamiglia@uab.cat}

\author{\aw{Andreas} Winter}
\affiliation{F\'isica Te\`orica: Informaci\'o i Fen\`omens Qu\`antics, Departament de F\'isica, Universitat Aut\`onoma de Barcelona, 08193 Bellaterra (Barcelona) Spain}
\affiliation{ICREA---Instituci\'o Catalana de Recerca i Estudis Avan\c{c}ats,
Pg.~Lluis Companys, 23, 08010 Barcelona, Spain}
\email{andreas.winter@uab.cat}

\date{18 August 2019}

\begin{abstract}
It is by now well-established that there exist non-local games for which the 
best entanglement-assisted performance is not better than the best classical 
performance. Here we show in contrast that any two-player XOR game, for which 
the corresponding Bell inequality is tight, has a quantum advantage. 
In geometric terms, this means that any correlation Bell inequality for which 
the classical and quantum maximum values coincide, does not define a facet, 
i.e. a face of maximum dimension, of the local (Bell) polytope. 
Indeed, using semidefinite programming duality, we prove upper bounds on the 
dimension of {these faces}, bounding it far away from the maximum. In 
the special case of non-local computation games, {it had been shown before
that they are not facet-defining;} our result generalises and improves this.
{As a by-product of our analysis, we find a similar upper bound on the
dimension of the faces of the convex body of quantum correlation matrices, 
showing that (except for the trivial ones expressing the non-negativity of 
probability) it does not have facets.}
\end{abstract}

\maketitle

\ssection{Introduction}
In 1964, Bell \cite{BellTheorem} proved that some predictions of quantum theory 
regarding the correlations between distant events cannot be explained by any 
\aw{classical, i.e. local realistic} theory. He derived a simple observable 
criterion that any classical theory must obey, and showed that particular 
measurements performed by two parties on a maximally entangled state could violate 
it. \aw{What we now call a Bell inequality was introduced in \cite{CHSHArticle}, 
as an upper bound on a single linear function of observable probabilities, i.e. 
an operational expectation value.}
This quantity has been experimentally measured \cite{aspect_experimental_1982,aspect_closing_2015} 
and \aw{shown} to exceed the classical upper bound, and thereby elevated Bell's 
theorem to one of the \aw{deepest} results in science, with a momentous impact on 
the way we understand the physical world. Quantum entanglement is responsible for 
\aw{these} observed correlations and it is also the key ingredient in most of the 
quantum \aw{informational advantage} in computation, communication, and sensing 
applications. Non-locality on its own has also been identified as a valuable resource 
in applications such as secure key distribution \cite{acin_device-independent_2007}, 
certified randomness \cite{pironio_random_2010}, 
reduced communication complexity \cite{buhrman_nonlocality_2010},
self-testing \cite{mayers_self_2004,supic_self-testing_2016},
and computation \cite{anders_computational_2009}.

In order to advance in the fundamental understanding of the perplexing features 
of non-local correlations and their technological spin-offs, \aw{in the last 
decades important efforts have been devoted} to their characterisation and 
exploitation (see \cite{brunner_bell_2014} for a recent review). 
Tsirelson \cite{cirelson_quantum_1980} computed the maximal violation of the 
CHSH inequality \cite{CHSHArticle} attainable by quantum mechanics;
later, Popescu and Rohrlich \cite{popescu_quantum_1994} 
(see also \cite{tsirelson_1993})
showed that although quantum correlations 
belong to the set of no-signalling correlations ---they do not allow for instantaneous 
communication---, they do not attain the full strength \aw{allowed in principle 
by the no-signalling condition}. 
These results reveal astonishing features \aw{of} the convex 
sets of classical ($\mathcal{C}$), quantum ($\mathcal{Q}$) and no-signalling 
($\mathcal{NS}$) correlations, in particular the strict inclusion 
$\mathcal{C}\subset \mathcal{Q}\subset \mathcal{NS}$. However, a lot remains to 
be understood, \aw{both} at the conceptual and mathematical level. For instance, 
the fact that $\mathcal{Q}\subset\mathcal{NS}$ spurred the search for underlying 
\aw{operational principles} that would single out quantum correlations 
\aw{among} general no-signalling 
ones \cite{pawlowski_information_2009,navascues_miguel_glance_2010,navascues_almost_2015}.
An approach that may assist in identifying such operationally defined principles 
and that may unveil new applications of non-locality is based on cooperative games 
of incomplete information \cite{cleve_Procconsequences_2004}, where two (or more)
remote parties cooperate to win a probabilistic game against a referee. Indeed, 
an increased winning probability when the two parties use quantum resources (quantum 
entanglement) instead of classical ones (which includes shared randomness), 
is equivalent to the violation of a Bell inequality. 
Gill \cite{gill_bell, krueger_open_2005}
asked the fruitful question whether all \emph{tight} Bell inequalities are violated 
by quantum mechanics. Here, tightness means that the inequality cannot be expressed 
as a positive linear combination of other Bell inequalities, or in geometric 
terms, that the Bell inequality defines a facet of the polytope of classical correlations 
(see below). 
Linden \textit{et al.} \cite{linden_quantum_2007} (motivated by \cite{PR-limits})
\aw{found the first class of two-player games, called} 
\emph{non-local computation (NLC)}, that \aw{have} no quantum advantage; 
the tightness of their Bell inequalities was posed as an open question 
in \cite{linden_quantum_2007}.
Almeida \textit{et al.} \cite{almeida_guess_2010} presented another case, 
\aw{the multi-party \emph{guess your neighbour's input (GYNI)}} game, 
\aw{shown to define} a tight Bell inequality without quantum violation. 
Around the same time, \aw{it was understood that NLC games never define facets of the 
Bell polytope \cite{winter_quantum_2010}, though this result was never written up; 
a proof was eventually published by Ramanathan \textit{et al.} \cite{ramanathan_tightness_2017}}. 

In this Letter we \aw{prove} that XOR games never define a facet of the Bell 
polytope (thus extending the result for NLC), answering Gill's question in the 
affirmative for the correlation polytope: all nontrivial tight correlation Bell 
inequalities have quantum violations.
The \aw{remainder of this Letter is} structured as follows: i) we first introduce 
the general formalism to describe the set of no-signalling, local classical and 
quantum correlations; ii) we briefly present the XOR games and give general expressions 
for the winning probabilities under different locality scenarios; iii) we present 
our main theorem and {the main ideas of} its proof; 
iv) we extend our result to the quantum set 
of correlations, and conclude with a discussion and outlook. 
In the Supplementary Material we {give the full proofs of our result, and 
present a simplified analysis in} the particular case of NLC, reconstructing the 
argument alluded to in \cite{winter_quantum_2010}, and improving 
\cite{ramanathan_tightness_2017} by giving a bound on the dimension of the face. 
 
\begin{figure}[t]
\centering
{\includegraphics[width=50mm]{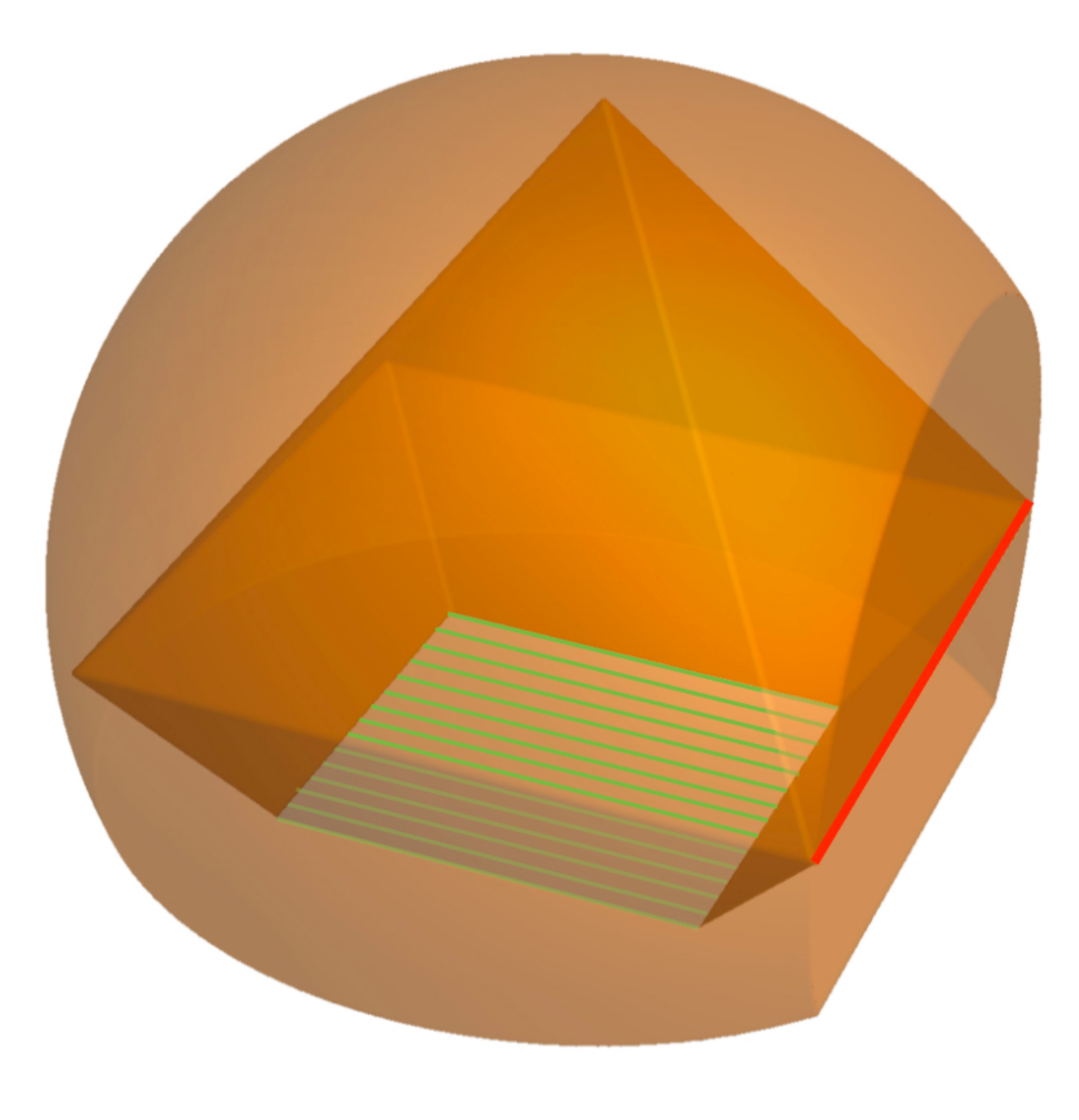}}
\caption{{Three-dimensional schematic} of a local \lle{correlation} polytope and the convex 
         body {of} quantum correlations, illustrating different types 
         of regions where local and quantum boundaries coincide (corresponding to XOR 
         games without  quantum advantage): the green-striped region is a \textit{facet} 
         (face of dimension $2$) of the Bell polytope, while the red line is a \textit{face} 
         (not a facet). 
         Theorem \ref{Our Theorem 1} excludes the former, and hence implies 
         that tight Bell inequalities (all facets) have quantum violations. 
         \lle{Theorem \ref{Our Theorem3} implies that all facets of the set of 
         quantum correlations are trivial}.}
\label{Punts}
\end{figure}

\medskip

\ssection{\aw{No-signalling behaviours}}
Consider a bipartite system where two parties, Alice and Bob, can perform 
measurements $x\in[m_A]$ and $y\in[m_B]$, respectively, obtaining the respective 
outcomes $a$ and $b$, which are binary. The event of obtaining $a$ and $b$ when 
the local measurements $x$ and $y$ are performed, is given according a 
conditional probability $p(a,b|x,y)$. 
In order to avoid instantaneous communication between the distant parties, 
forbidden by special relativity, this probability $p$ must satisfy the 
\aw{\emph{no-signalling property}}, i.e. $\sum_{b}p(a,b|x,y) = \sum_{b}p(a,b|x,y') \quad \forall a, x, y, y'$
and analogously summing Alice's outcomes, which {in physical terms excludes} 
that any party signals to another party by their choice of input. 

The set of all probabilities satisfying the above no-signalling property, called the 
no-signalling set ($\mathcal{NS}$), constitutes a polytope of dimension 
$D=m_Am_B+m_A+m_B$ \cite{pironio_lifting_2005}. 
In an attempt to explain phenomena locally, one may \aw{consider} the existence 
of classical local (hidden) variables $\lambda\in\Lambda$, distributed according 
a probability law $\rho(\lambda)$, such that the probability of an observed event 
can be written as $
  p(a,b|x,y) = \int_{\Lambda}d\lambda \rho(\lambda)p(a|x,\lambda)p(b|y,\lambda).
$
The set of all probabilities of this form is called the local 
set, which is also a polytope, the so-called Bell or local polytope, of the same dimension 
\aw{as} the no-signalling polytope \cite{pironio_lifting_2005}, and denoted $\mathcal{C}$. 
A  polytope 
\aw{$\cP$ can equivalently be defined as the convex hull of a finite set of points,
$\cP = \conv\{v_j : j=1,\ldots,k\}$,
or as a bounded intersection of finitely many closed half-spaces,
$\cP = \{\vec{v} : \forall i=1,\ldots,\ell\ \vec{u}_i\cdot\vec{v}\leq w_i\}$ \cite{ConvexPolytopes}. 
A linear inequality for $\cP$ is a $\vec{u}\cdot\vec{v}\leq w$ that holds
for all $\vec{v}\in\cP$; in geometry, $H = \{\vec{v} : \vec{u}\cdot\vec{v} = w\}$ 
is also called a \emph{supporting hyperplane} of $\cP$.
For a given supporting hyperplane, the set of point $\vec{v}\in\cP$ achieving the 
equality is called a \emph{face} of the polytope $\cP$, $\cF = \cP \cap H$.}
In the case of $\mathcal{C}$, its corresponding inequalities are precisely the 
Bell inequalities. The faces of maximum dimension $D-1$ are called facets; when 
dealing with $\mathcal{C}$, the corresponding inequalities are called tight,
or facet, Bell inequalities.	
Facet inequalities give the minimal characterisation of the polytope in terms of 
half-spaces in the sense that any other inequality that holds for the polytope 
can be written as a \aw{positive} linear combination of the facet inequalities.

To state and prove our results, we use a convenient minimal parametrisation 
of the no-signalling polytopes.
Any no-signalling \emph{behaviour} $p(a,b|x,y)$ is fully characterised 
by the first moments \mbox{$\alpha_{x}=\mean{(-1)^{a}}_{x,y}=\sum_{a,b} (-1)^{a} p(a,b|x,y)$}
and \mbox{$\beta_{y}=\mean{(-1)^{b}}_{x,y}=\sum_{a,b} (-1)^{b} p(a,b|x,y)$} (which,  
due to {no-signaling property}, are independent of $y$ and $x$, resp.);
and the correlators $c_{xy}=\mean{(-1)^{a+b}}_{x,y}=\sum_{a,b} (-1)^{a+b} p(a,b|x,y)$. 
Indeed, from these $D=m_Am_B+m_A+m_B$ values we recover
$4 p(a,b|x,y)= 1+(-1)^{a} \alpha_{x}+ (-1)^{b} \beta_{y}+(-1)^{a+b} c_{xy}$.
The polytope of $\mathcal{NS}$ distributions can hence be described by the tuple
$(\ket{\alpha},\ket{\beta},C) \in \R^D$, 
where $\ket{\alpha} = \sum_{x} \alpha_{x}\ket{x} \in \R^{m_A}$
and $\ket{\beta} = \sum_{y} \beta_{y}\ket{y} \in \R^{m_B}$ 
are the local moment vectors,
and $C= \sum_{x,y} c_{xy}\ketbra{x}{y}$ is the correlation matrix.
The local, or Bell, polytope arises when restricting the strategies to convex 
combinations of local deterministic ones \cite{fine_hidden_1982}. 
We use the subscript $c$ to label such extremal classical strategies 
$\ket{\alpha_c}\in\{-1,1\}^{m_A}$, $\ket{\beta_c}\in\{-1,1\}^{m_B}$, for which 
we note that $c_{xy}=\alpha_{x}\beta_{y}$, that is $C=|\alpha_{c}\rangle\!\langle\beta_{c}|$. 
The Bell polytope is then given by the convex hull
\begin{equation}
\label{1st span}
  \cC 
   = \conv\left\{\left(\ket{\alpha_{c}},\ket{\beta_{c}},|\alpha_{c}\rangle\!\langle\beta_{c}|\right)\right\}.
\end{equation}

\aw{Finally, there are at least two definitions of sets of quantum behaviours 
that we have to consider: the most general setting is of a state $\ket{\psi}$ in a
Hilbert space $\cH$, together with Alice's and Bob's observables $\hat{a}_{x}$ 
and $\hat{b}_{y}$, respectively, with eigenvalues $0$ and $1$ (i.e. they
are projectors), and such that for all $x$, $y$, $[\hat{a}_{x},\hat{b}_y]=0$.
Then, $p(a,b|x,y)=\bra{\psi}\hat{a}_{x}\hat{b}_y\ket{\psi}$, and hence in the
above parametrisation for $\mathcal{NS}$, we have  
$\alpha_{x}=\bra{\psi} (-1)^{\hat{a}_{x}} \ket{\psi}$, 
$\beta_{x}=\bra{\psi}(-1)^{\hat{b}_{y}} \ket{\psi}$ and 
$c_{xy}=\bra{\psi} (-1)^{\hat{a}_{x}}(-1)^{\hat{b}_{y}} \ket{\psi}$. 
The set of such behaviours is denoted $\cQ_{\text{com}}$, the subscript standing for
``commuting'' strategies, and it is known to be a closed convex set contained in
$\mathcal{NS}$.
The other, traditionally considered setting is that $\cH=\cH_A\otimes\cH_B$
is a tensor product Hilbert space, and that $\hat{a}_{x} = \hat{a}^A_{x}\otimes\1_B$
and $\hat{b}_{y} = \1_A\otimes\hat{b}^B_{y}$, with observables $\hat{a}^A_{x}$
on $\cH_A$ and $\hat{b}^B_{y}$ on $\cH_B$. The corresponding set of behaviours
is convex, but recently has been shown not to be closed \cite{Slofstra}
for $m_A,m_B\geq 5$ \cite{Paulsen}, which is why we define $\cQ_\otimes$ to be its 
closure. By definition, $\cQ_\otimes \subseteq \cQ_{\text{com}}$, and while it is open 
whether the two sets are equal, this would be equivalent to Connes' long-standing
Embedding Problem in the theory of von Neumann algebras \cite{junge_connes_2011,ozawa_about_2004}.}
The sets of quantum behaviours are convex sets, but unlike the classical and 
no-signalling sets, they are not polytopes: they have uncountably many 
extremal points, and part of their boundary is curved.

The study of non-local correlations is often \aw{carried out in a} simplified 
scenario, the so-called \emph{correlation polytope}, which is given by the set of 
correlators $C$ (without including the local terms). The corresponding linear 
criteria that define the \aw{set of classical/local correlations} $\cC_0$ are 
called \emph{correlation Bell inequalities} 
\cite{Froissart81,fine_hidden_1982,werner_wolf_2001}. 
The projection of quantum and no-signalling behaviours 
onto the correlator subspace are {the quantum  $\cQ_0$ and no-signalling $\mathcal{NS}_0$ correlations respectively.
Note that by Tsirelson's results \cite{cirelson_quantum_1980}, both $\cQ_{\text{com}}$
and $\cQ_\otimes$ give rise to the same quantum correlator set, realized in fact
with local Hilbert spaces $\cH_A$ and $\cH_B$ of bounded dimension.

\medskip
\ssection{XOR games}
Non-local  games provide an {intuitive operational} setting in which to 
cast Bell inequalities, and relate those to the {well-established} field of interactive 
proofs in computer science. Here we will focus on {the particular class of two-player}
XOR games \cite{cleve_Procconsequences_2004}, where the outcomes of each party are 
binary and the winning condition depends on the exclusive disjunction (XOR) of the outcomes. 
XOR games have a prominent role in non-locality: the paradigmatic CHSH 
inequality \cite{CHSHArticle}, the GHZ paradox \cite{greenberger_bells_1990} 
and NLC \cite{linden_quantum_2007} can all be phrased as XOR games; and most importantly, they 
provide a characterisation of the correlation Bell polytope as will become
apparent below.

\begin{figure}[ht]
  \centering
  \includegraphics[width=55mm]{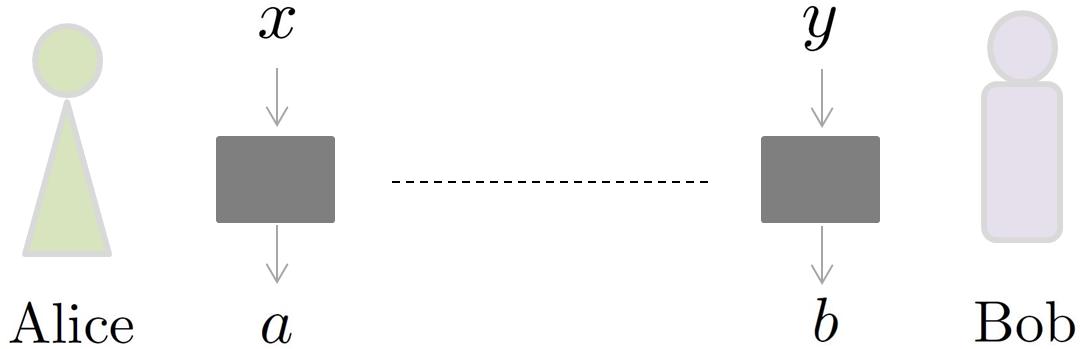}
  \caption{Representation of a XOR game. The goal is that Alice and Bob output 
           $a$ and $b$ such that $a\oplus b=f(x,y)$.}
  \label{NLC game figure}
\end{figure}

In an XOR game, see Fig. \ref{NLC game figure}, the referee provides queries 
$x\in [m_{A}]$ to one player (Alice) and $y\in [m_{B}]$ to the second player (Bob), 
with the promise that queries are sampled from some prior probability distribution 
$q(x,y)$ known to both players. In order to win the game, upon receiving their 
inputs $x$ and $y$, Alice and Bob must produce a binary output $a,b\in\{0,1\}$, 
respectively, such that $a\oplus b=f(x,y)$,
where $f$ is a given Boolean function also known to both players. The performance 
of Alice and Bob's strategy is quantified by the average winning probability 
\begin{equation}
  \label{w def}
  \omega = \sum_{x,y} q(x,y)p\!\left(a\oplus b\!=\!f(x,y)|x,y\right)= \frac12 (1+\xi), 
 \end{equation}
where $\xi = \sum_{x,y} q(x,y)(-1)^{f(x,y)} c_{xy}$
%
is the gain {(or bias)}. Note that XOR games can always be won with at least 
probability $\frac12$ if Alice (or Bob) produces a random output independently 
of the input. Since  $q(x,y)$ and $f(x,y)$ are a given, we can characterise 
the game by the \aw{so-called} game matrix
\mbox{$
  \Phi = \sum_{x,y}(-1)^{f(x,y)}q(x,y)|x\rangle\!\langle y|,
$}
so that the gain can be written in terms of the 
correlation matrix as $\xi=\tr{C\Phi^{T}}$. \aw{Every correlation Bell inequality,
as it is based on a linear function of the correlators $C$, can be written
in the form $\tr{C\Phi^{T}} \leq \xi$, and by rescaling if necessary,
$\Phi$ can be chosen as the game matrix of a suitable XOR game.} 
The optimal classical success probability can always be attained by 
extremal \aw{(i.e. deterministic)} strategies $|\alpha_c\rangle$ and $\ket{\beta_c}$,  
with $C=|\alpha_{c}\rangle\!\langle\beta_{c}|$. 
Hence the gain of the local classical average winning probability can be written as:
\begin{equation}
  \label{eq:cl-gain}
  \xi_{c} = \max_{\alpha_{c},\beta_{c}} {\langle\alpha_{c}|}\Phi|\beta_c\rangle.
\end{equation}
In the Supplementary Material (A) we present various useful ways to 
write the quantum gain, which are employed in the proofs of our main results.

It is easy to see that no-signalling behaviours  allow to win XOR games 
with $\omega_{NS}=1$ \cite{ramanathan_characterizing_2014}, and therefore any XOR 
game with $\omega_{c}<1$ will correspond to a nontrivial Bell inequality, \aw{i.e.
one that can potentially be violated quantumly}. See the Supplementary Material
(B) for further discussion of this point.

\aw{Observe finally that without loss of generality, we may restrict ourselves 
to XOR games} with game matrices $\Phi$ that have no all-zero rows or columns. Indeed, 
because such a row or column of zeros imply that the marginal $q(x)$ or $q(y)$ 
are zero for some inputs, we can redefine the set of possible queries 
(decreasing $m_{A}$ or $m_{B}$ accordingly) to obtain an equivalent game without 
all-zero rows or columns in its game matrix. We refer to such games with 
$q(x)>0$ and $q(y)>0$ for all $x\in[m_{A}]$ and $y\in[m_{B}]$ as \emph{exhaustive games}.



\medskip

\ssection{Results}
\aw{For a long time, it was implicitly assumed that if Alice and Bob use entangled 
strategies,} they can attain a greater success probability than if they are limited 
to classical resources, \aw{for any nontrivial Bell inequality.} 
As explained in the introduction, \aw{it took a while to find examples of nontrivial games 
that do not show any quantum advantage}. 

Here we show that XOR games (which characterise the correlation polytope) 
without quantum advantage never define a facet of the Bell polytope 
({full behaviours or correlations}).
This in turn implies that all (nontrivial) tight correlation Bell inequalities 
have quantum violations.

In  \cite{ramanathan_characterizing_2014} Ramanathan \textit{et al.}
 derived a necessary and sufficient condition for a two-player XOR game to have 
no quantum advantage, which will turn out to be fundamental for the proof of our first
 result. 

\begin{theorem}
  \label{Our Theorem 1} 
  If an exhaustive XOR game has no quantum advantage, 
  the corresponding Bell inequality \aw{does not define a facet of the 
  Bell polytope, nor of the correlation Bell polytope}.
\end{theorem}

{The \emph{proof} -- see the Supplementary Material (C) for full details -- 
proceeds by bounding the dimension of the face $\mathcal{F}$ in the Bell polytope 
corresponding to the maximum classical  bias $\xi_{c}$ of the given 
XOR game: 
\[\begin{split}
  \mathcal{F} &= \{ (\ket{\alpha},\ket{\beta},C) \in\cC : \tr C\Phi^T = \xi_{c} \} \\
              &= \conv\left\{ (\ket{\alpha_c},\ket{\beta_c},\ketbra{\alpha_c}{\beta_c}) 
                                    : \bra{\alpha_c}\Phi\ket{\beta_c} = \xi_{c} \right\}. 
\end{split}\]
The first ingredient is the characterisation of the maximum quantum bias $\xi_Q$
by semidefinite programming (SDP) \cite{Wehner}, which by SDP duality leads to 
strong constraints on any optimal strategy via complementary slackness. 
From the assumption that $\xi_Q=\xi_c$, this leads to the second, and key, 
insight of the proof, namely that in any pair $(\ket{\alpha_c},\ket{\beta_c})$
of optimal classical strategies, Alice's and Bob's local answers uniquely 
determine each other linearly: as we show in the proof, $\ket{\beta_c}=F\ket{\alpha_c}$
for a certain matrix $F$. 
Assuming w.l.o.g. $m_{A}\leq m_{B}$, we thus have
\begin{equation}
  \mathcal{F} = \conv\left\{(\ket{\alpha_{c}}\!,F\ket{\alpha_c}\!,\ket{\alpha_c}\!\!\bra{\alpha_{c}}F^T)
                                                                  : \ket{\alpha_c} \text{ opt.} \right\}, \nonumber
\end{equation}
and its dimension can be upper bounded by that of
\[
  \aff\left\{(\ket{\alpha_{c}}\!,F\ket{\alpha_c}\!,\ket{\alpha_c}\!\!\bra{\alpha_{c}}F^T)
                                              : \ket{\alpha_c} \in \{\pm 1\}^{m_A} \right\}, 
\]
which is $m_A+\frac12 m_A(m_A-1)<D-1$. In the case of the correlation polytope,
the dimension is similarly upper bounded by $\frac12 m_A(m_A-1)<m_{A} m_{B}$.

\begin{theorem}
  \label{Theorem2}
  All nontrivial tight correlation Bell inequalities for bipartite 
  systems with binary outcomes have a quantum violation. 
\end{theorem}

\begin{proof}
{Consider a (non-exhaustive) XOR game with $M_{A}$ (Alice) and $M_{B}$
(Bob) inputs. W.l.o.g. the first $m_A$ ($m_B$) inputs of Alice (Bob) have
non-zero probability, the rest are never asked, so we can apply 
Theorem \ref{Our Theorem 1} to the reduced exhaustive game, which relates 
the optimal strategies for the  indices $x\in[m_A]$ and $y\in[m_B]$, but 
leaves completely unconstrained the remaining ones. Hence, given a strategy by Alice 
$|\alpha_{c}\rangle\oplus|\alpha_{c}'\rangle$, 
Bob's  strategy must be 
$(F|\alpha_{c}\rangle)\oplus|\beta_{c}'\rangle$, 
where $|\alpha_{c}'\rangle \in \{\pm 1\}^{M_{A}-m_{A}}$ and 
$|\beta_{c}'\rangle \in \{\pm 1\}^{M_{B}-m_{B}}$. 

We thus arrive at a codimension $\Delta = D-\dim\mathcal{F}$ of the face of
$\Delta \geq m_{B} + M_A(m_B-m_A) + \frac{m_A}{2} (m_A+1)>1$. That is, XOR games with quantum equal to classical value 
do not define a facet of the full Bell polytope. 
Following the same argument for the correlation polytope leads to the 
codimension 
  $\Delta_0 \geq M_A(m_B-m_A) + \frac12 m_A(m_A+1)$,
and this is greater than $1$ unless $m_A=m_B=1$, corresponding precisely to the 
trivial inequalities $|c_{xy}|\leq 1$.}
\end{proof}

{In the Supplementary Material (D) we give a different proof of the 
above result for non-local computation games, showing also that the 
dimension bounds are asymptotically attained for nontrivial games.}


From the proof of Theorem \ref{Our Theorem 1} we learn that the optimal 
extremal behaviours in an exhaustive XOR game with no quantum advantage are 
fully determined by the strategy of one of the parties. We will now show 
that this feature actually extends to all optimal quantum behaviours of arbitrary
XOR games.

To understand the following theorem, we recall the definition of a face 
$\cF\subset \cQ$ of a general compact convex set $\cQ$: namely, 
that whenever $\cF\ni \vec{p} = t\vec{q}+(1-t)\vec{r}$,
$0<t<1$, then both $\vec{q},\,\vec{r} \in \cF$. An \emph{exposed face} is 
obtained as $\cF = \cQ \cap H$ with a supporting hyperplane $H$ of $\cQ$; all 
exposed faces are faces of $\cQ$, but not vice versa. However, for polytopes 
every face is an exposed face \cite{Rockafeller,ConvexPolytopes}. 
Also, facets, and more generally maximal faces, are always exposed.

Note that this result explains the previous two theorems on the classical 
behaviours as being due to broader properties of the quantum sets. 

\begin{theorem}
  \label{Our Theorem3} 
  Nontrivial XOR games, or equivalently nontrivial correlation 
  Bell inequalities, never define a facet of the quantum sets of behaviours 
  $\cQ_{\text{com}}$ and $\cQ_\otimes$.
  As a consequence, the set $\cQ_0$ of quantum correlations has no
  nontrivial facets.
\end{theorem}

{See the Supplementary Material (C) for the complete proof. 
To give the broad outline, we start with an exhaustive XOR game. 
The complementary slackness condition in the proof of Theorem \ref{Our Theorem 1}
for the optimal quantum strategy leads to 
$|\beta_{y}\rangle=\sum_{x'} F_{yx'}|\alpha_{x'}\rangle$,
with the same matrix $F$ as before. 
In other words, once again Alice's optimal quantum strategy uniquely determines Bob's,
and vice versa.

Thus, we get for an optimal quantum correlation matrix 
$C_{xy}=\braket{\alpha_{x}}{\beta_{y}}=\sum_{x'} F_{yx'}\braket{\alpha_{x}}{\alpha_{x'}}$, 
and hence the dimension of their affine span is bounded by that of the Gram 
matrices $[\braket{\alpha_{x}}{\alpha_{x'}}]_{xx'}$, with dimension $\frac12 m_A(m_A-1)$, 
leading to the same dimension bound as in the proof of Theorem \ref{Our Theorem 1}. 
Non-exhaustive XOR games are treated as in the proof of Theorem \ref{Theorem2}.}

\medskip
\ssection{Discussion and outlook} 
We have shown that a two-party correlation Bell inequality (XOR game) with 
no quantum violation (quantum advantage) cannot define a facet of the Bell 
polytope. The contrapositive of this statement has deep physical implications: 
all tight correlation Bell inequalities exhibit a quantum violation. 
{In fact, we have proven lower bounds on the codimension of the defined face
(increasing with the number of inputs). As a consequence, when the codimension 
is lower bounded by $\Delta>1$, not only all tight correlation inequalities 
will have quantum violations, but also those corresponding to faces $\mathcal{F}$ 
with $\dim\mathcal{F}\geq D-\Delta$.}  

On the way, we have proved that this in fact is due to a broader property of 
the convex set of quantum correlations, namely that it does not have any non-trivial
facets, only lower-dimensional faces. It remains to be seen what the physical meaning
of this curious geometric observation is.

Our results appear very much tied to the world of two-player, binary outcome
XOR games. This leaves open the questions whether XOR games for more than two players 
can define common facets of the quantum and classical sets (note that GYNI
defines such a facet, but it is not an XOR game), and whether for two players 
there are any tight Bell inequalities at all without quantum violations. It might be
possible to extend our results at least two two-players MOD-$q$ games where
each player has a $q$-ary outcome.

\ssection{Acknowledgements}
We thank Mafalda Almeida, Nicolas Brunner and Paul Skrzypczyk for discussions on 
non-local computation, and Dardo Goyeneche for rekindling our interest in Bell 
inequalities without quantum violation. 
Support from the Spanish MINECO, project FIS2016-80681-P with the support of 
AEI/FEDER funds, and from the Generalitat de Catalunya, project CIRIT 2017-SGR-1127,
is acknowledged.

\bibliographystyle{unsrt}
\bibliography{references}

\vspace{1cm}

\section{Supplementary Material}

\subsection{A. Quantum bias}
\aw{If Alice and Bob use quantum resources, i.e. a (possibly) entangled state $|\psi\rangle\in\mathcal{H}$, 
with with observables $\hat{a}_x$ and $\hat{b}_y$ depending on their respective 
inputs $x$ and $y$, to produce their measurement outcomes $a$ and $b$, then,
using the expression given in the main text for the quantum correlations $c_{xy}$, the quantum gain 
can be written as 
\begin{equation}\begin{split}
  \xi_{Q} &= \sum_{x,y} q(x,y) \bra{\psi}(-1)^{\hat{a}_x+\hat{b}_y+f(x,y)}\ket{\psi} \\
          &= \sum_{x,y} q(x,y)(-1)^{f(x,y)} \bra{\alpha_x}\beta_y\rangle, 
\end{split}\end{equation}
with $\ket{\alpha_x} = (-1)^{\hat{a}_x}\ket{\psi}$ and
$\ket{\beta_y} = (-1)^{\hat{b}_y}\ket{\psi}$ unit vectors in $\cH$.
Vice versa, given any set of unit vectors in any complex Hilbert space,
there exists an equivalent set, i.e. with the same pairwise inner 
products $\bra{\alpha_x}\beta_y\rangle$, of the above form, on a tensor
product Hilbert space \cite{cirelson_quantum_1980}.}
It will \aw{prove} convenient to define the following states on an extended vector space,
\begin{equation}\begin{split}
  \label{Def quantum alpha beta and Phi}
  |\alpha\rangle_{Q} &= \sum_x (-1)^{\hat{a}_x}\ket{\psi} \otimes \ket{x},\\
  |\beta\rangle_{Q}  &= \sum_y (-1)^{\hat{b}_y}\ket{\psi} \otimes \ket{y}. 
\end{split}\end{equation}
With this, the expression for the optimal quantum gain reads
\begin{equation}
  \label{omegaQ}
  \xi_{Q} 
    = \max_{\ket{\alpha}_{Q},\ket{\beta}_{Q}} {_{Q}\langle\alpha|}\mathbb{I}\otimes\Phi|\beta\rangle_{Q},
\end{equation}
which resembles the expression \eqref{eq:cl-gain} for its classical counterpart.

\subsection{B. XOR games can always be won using no-signalling behaviours}
It is easy to understand the previous observation of $\omega_{NS}=1$ for 
every XOR game, by realising that an XOR game can always be won with certainty 
using a behaviour that outputs \emph{uniformly random} local bits $a$ and $b$,
and so is evidently no-signalling:
\begin{equation}
  p(a,b|x,y) = \begin{cases}
                 \frac12 & \text{ if } a\oplus b = f(x,y), \\
                 0       & \text{ otherwise}.
               \end{cases}
\end{equation}
This behaviour has $\ket{\alpha}=\ket{\beta}=0$, but its correlator 
term is remarkably $c_{xy} = (-1)^{f(x,y)}$, i.e. an arbitrary $\pm 1$-matrix.
In geometric terms, this says that $\mathcal{NS}_0$ is the $\ell_\infty$ unit 
ball (aka hypercube), in other words $\mathcal{NS}_0$ is entirely characterised 
by the inequalities $-1\leq c_{xy} \leq 1$. 
These are indeed the trivial Bell inequalities, since they follow from the 
non-negativity of the probabilities $p(a,b|x,y)$. 

To back up the definite article in the previous statement, it is in fact
known that they not only define facets of $\mathcal{NS}_0$, but also of 
$\cC_0$ and thus of $\cQ_0$; namely, for every pair $(x_0,y_0)$, the 
behaviour $C_0 = \ketbra{x_0}{y_0}$ is in $\cC_0$, and so is an entire 
sufficiently small neighbourhood of $C_0$ of correlators $C$ in the supporting 
hyperplane $H = \{C : C_{x_0y_0} = 1\}$ (and analogously for $-C_0$). 
Again, this is easy to see: consider local classical strategies 
$\ket{\alpha_c}=\sum_x \alpha_x\ket{x}$ and $\ket{\beta_c}=\sum_y \beta_y\ket{y}$
with $\alpha_{x_0}=\beta_{y_0}=1$, so that we can write
$\ket{\alpha_c}=\ket{x_0} + \ket{\alpha'}$ and $\ket{\beta_c}=\ket{y_0} + \ket{\beta'}$.
Then, 
\begin{equation}
  \ketbra{\alpha_c}{\beta_c} 
      = C_0 + \ketbra{x_0}{\beta'} + \ketbra{\alpha'}{y_0} + \ketbra{\alpha'}{\beta'}.
\end{equation}
By taking convex combinations over 
$\ket{\beta'} = \epsilon\ket{y_1} + \sum_{y\neq y_0,y_1} \beta_y\ket{y}$,
with uniformly random $\beta_y=\pm 1$ ($y\neq y_0,y_1$) and 
uniformly random $\ket{\alpha'}$, we annihilate all terms except one,
showing that $C_0+\epsilon\ketbra{x_0}{y_1} \in \cC_0$, for $\epsilon=\pm 1$. 
Similarly, we find $C_0+\epsilon\ketbra{x_1}{y_0} \in \cC_0$. 
Finally, by taking convex combinations over 
$\ket{\alpha'}=\sum_{x\neq x_0} \alpha_x\ket{x}$ and 
$\ket{\beta'}=\epsilon\alpha_{x_1}\ket{y_1} + \sum_{y\neq y_0,y_1} \beta_y\ket{y}$,
with uniformly random $\ket{\alpha'}$ and uniformly random 
$\beta_y=\pm 1$ ($y\neq y_0,y_1$), we get that 
$C_0+\epsilon\ketbra{x_1}{y_1} \in \cC_0$. The convex hull of all these
points, a hyper-octahedron, contains a small neighbourhood of
$C_{0}$ in the hyperplane $H$.

\subsection{C. Proofs}
Here, we give the complete proofs of Theorems \ref{Our Theorem 1}
and \ref{Our Theorem3} in the main text.

\begin{proofof}[of Theorem \ref{Our Theorem 1}]
Our aim is to bound the dimension of the face $\mathcal{F}$ defined by the 
maximum classical bias $\xi$ of a given XOR game:
\begin{equation}\begin{split}
  \mathcal{F} &= \{ (\ket{\alpha}\!,\ket{\beta}\!,C) \in\cC : \tr C\Phi^T = \xi_{c} \} \\
              &= \conv\left\{ (\ket{\alpha_c}\!,\ket{\beta_c}\!,\ketbra{\alpha_c}{\beta_c}) 
                                          : \bra{\alpha_c}\Phi\ket{\beta_c} = \xi_{c} \right\}. 
\end{split}\end{equation}
In particular, we want to show that its dimension $\dim\mathcal{F}$ is strictly 
lower than the dimension of a facet, $D-1$, or equivalently that the codimension of the 
face in the classical polytope is $\Delta=D-\dim\mathcal{F}>1$.


We start by considering the consequences of assuming no quantum advantage, 
i.e. $\xi_{c}=\xi_{Q}$. The optimal quantum gain can be written as 
$\xi_{Q}=\max_{C}\tr{C\Phi^T}$, where 
$C=\sum_{x,y}\langle\alpha_x|\beta_y\rangle\ketbra{x}{y}$ is the quantum correlator
matrix, which is fully characterised \cite{cirelson_quantum_1980} by the inner products 
of an arbitrary set of unit vectors in $\R^{m_{A}+m_{B}}$. In terms of the 
characterisation in \eqref{Def quantum alpha beta and Phi}, 
$\ket{\alpha_{x}}=\braket{x}{\alpha}_{Q}=(-1)^{\hat{a}_x}|\psi\rangle$, 
and analogously for Bob's strategy. In order to write this optimisation problem as a 
semidefinite program \cite{cleve_Procconsequences_2004,Wehner} we define the following 
Gram matrix  
\begin{equation}
  \Tilde{Q}=\left(\begin{array}{c|c}
                    R   & C\\
                    \hline
                    C^T & S  \\
            \end{array}\right),
\end{equation}
with $R_{xx'}=\langle\alpha_x|\alpha_{x'}\rangle$ and 
$S_{yy'}=\langle\beta_{y}|\beta_{y'}\rangle$.  
An equivalent characterisation of a Gram matrix of unit vectors is 
$\Tilde Q \succeq 0$ and $\tilde Q_{ii}=1$ for all $i\in[m_A+m_B]$, 
and consequently we can write the maximum quantum gain as:
\begin{equation}\begin{split}
  \label{eq:xi-primal}
  \xi_{Q} &= \max_{\tilde Q} \tr{\Tilde{Q}\Tilde{\Phi}},\\
          &  \text{ s.t. } \diag(\Tilde{Q})=(1,..,1),\, \Tilde{Q}\succeq 0, 
\end{split}\end{equation}
where 
$\Tilde{\Phi}=\frac{1}{2}\left(\begin{array}{c|c}
                                 0 & \Phi\\
                                 \hline
                                 \Phi^T & 0  \\
                          \end{array}\right)$. 
	
Consider now the Lagrangian, 

\begin{equation}
\begin{split}
  L &= \tr{\Tilde{Q}\Tilde{\Phi}}-\sum_i t_i\left(\tr{|i\rangle\!\langle i|\Tilde{Q}}-1\right) \\
    &=  \tr\left[\Tilde{Q}(\Tilde{\Phi}-\sum_i t_i|i\rangle\!\langle i|)\right]+\sum_i t_i,
\end{split}
\end{equation}
where $t_i\in\R$ are the Lagrange multipliers. Therefore,
\begin{equation}
  \max_{\Tilde{Q}\succeq 0} L
       = \begin{cases}
           +\infty    & \text{ if } \sum_i t_i|i\rangle\!\langle i|-\Tilde{\Phi}\not\succeq 0, \\
           \sum_i t_i & \text{ if } \sum_i t_i|i\rangle\!\langle i|- \Tilde{\Phi}\succeq 0,
         \end{cases}
\end{equation}
and thus the original SDP \eqref{eq:xi-primal} can be written in its dual form,
\begin{equation}
  \label{dual problem}
  \min \sum_i t_i \text{ s.t.} \sum_i t_i |i\rangle\!\langle i| \succeq \Tilde{\Phi}. 
\end{equation}
This holds because the primal and dual SDPs satisfy the condition for
strong duality.

From the above construction it follows that $\tr{\Tilde Q\Tilde{\Phi}} \leq \xi_{Q} \leq \sum_i t_i$
for any pair of primal and dual feasible solutions. Furthermore, by strong duality,
the maximum primal value $\xi_Q$ equals the solution of the dual (\ref{dual problem}).
That is, the optimal value is attained if and only if $\tr{\Tilde Q\Tilde{\Phi}}= \sum_i t_i$, 
or equivalently if the spurious term in the Lagrangian vanishes:  
$\tr\left[\tilde Q\left(\sum_it_i|i\rangle\langle i|-\Tilde{\Phi}\right)\right]=0$ 
(complementary slackness). Note that by the form of the dual problem (\ref{dual problem}), 
the $t_i$ are non-negative numbers; indeed, by our assumption that $\Phi$ has 
no all-zero rows or columns, we even can conclude that all $t_i > 0$.

Now, our hypothesis of no quantum advantage implies that $\xi_Q$ coincides with its 
classical counterpart 
$\xi_{c}={\langle}\beta_{c}|\Phi|\alpha_{c}\rangle = \tr\ketbra{s}{s}\Tilde{\Phi}^T$,
where we have defined $|s\rangle=|\alpha_{c}\rangle\oplus|\beta_{c}\rangle$. 
In other words, $\xi_{Q}$ can be reached with a classical correlation matrix
$\Tilde Q=\ketbra{s}{s}$. 
Letting $\Gamma = \sum_i t_i|i\rangle\!\langle i| = \frac12 \Sigma\oplus\Lambda \succ 0$,  
where $\Sigma$ and $\Lambda$ are diagonal positive definite matrices, the slackness 
condition reads $\tr\left[(\Gamma-\Tilde{\Phi})\ketbra{s}{s}\right]=0$, which in turn 
implies that $(\Gamma-\Tilde{\Phi})\ket{s}=0$, since both $\Gamma-\Tilde{\Phi}$ 
and $\ketbra{s}{s}$ are positive semidefinite. 
Therefore, $\Sigma|\alpha_{c}\rangle=\Phi|\beta_{c}\rangle$ and 
$\Lambda|\beta_{c}\rangle=\Phi^T|\alpha_{c}\rangle$, or equivalently 
\begin{equation}
  \label{best strategies}
  |\alpha_{c}\rangle=\Sigma^{-1}\Phi|\beta_{c}\rangle, \, 
  |\beta_{c}\rangle=\Lambda^{-1}\Phi^T|\alpha_{c}\rangle =: F |\alpha_{c}\rangle.
\end{equation}
This means that once Alice's optimal strategy is fixed, Bob's best strategy is 
\aw{uniquely} determined by her choice and vice versa. Assuming w.l.o.g. $m_{A}\leq m_{B}$, 
the dimension of the face $\mathcal{F}$ generated by such strategies is then
\begin{align}
  \label{eq:span1}
  \dim\mathcal{F} 
    &=    \dim \aff\left\{(\ket{\alpha_{c}}\!,F\ket{\alpha_c}\!,\ket{\alpha_c}\!\!\bra{\alpha_{c}}F^T)
                          : \ket{\alpha_c} \text{ opt.} \right\} \nonumber \\
    &\leq \dim \aff\left\{(\ket{\alpha_{c}}\!,\ket{\alpha_c}\!,\ket{\alpha_c}\!\!\bra{\alpha_{c}})
                          : \ket{\alpha_c}\in\{\pm 1\}^{m_A} \right\} \nonumber \\
    &=    m_A + \frac{m_A(m_A-1)}{2},
\end{align}
where $\aff$ denotes the affine span,  and where the second term in the dimension
comes from the fact that the affine span of the matrices $\ket{\alpha_{c}}\!\bra{\alpha_{c}}$
consists precisely of the real symmetric matrices with $1$'s along the diagonal.
This leads to a codimension of the face in the Bell polytope of 
\mbox{$\Delta \geq m_B + m_Am_B - \frac12 m_A(m_A-1) > 1$.} 

Similarly, we can bound the codimension of 
this facet in the correlation polytope, which is solely generated 
by the correlators $\ket{\alpha_{c}}\!\bra{\alpha_{c}}F^T$, leading to 
$\Delta_0 \geq m_Am_B - \frac12 m_A(m_A-1) > 1$, unless $m_A=m_B=1$,
which is a trivial case.

Therefore, we see that the defined face is not a facet in either setting.
\end{proofof}

\begin{proofof}[of Theorem \ref{Our Theorem3}] 
Starting from the definition $|q\rangle=|\alpha\rangle_{Q}\oplus|\beta\rangle_{Q}$,
we can write the quantum bias as 
$\xi_{Q} = {_{Q}\!{\langle}}\beta|\Phi|\alpha\rangle_{Q}
         = \langle q|\1\otimes\Tilde{\Phi}|q\rangle$.
It is straightforward to check, following the steps in the proof of 
Theorem \ref{Our Theorem 1}, that the complementary slackness condition translates 
into $(\1\otimes\Gamma-\1\otimes\Tilde{\Phi})\ket{q}=0$, which in the case of 
exhaustive games leads to 
\begin{equation}
  \label{best strategies q}
  |\beta\rangle_{Q} = (\1 \otimes F) |\alpha\rangle_{Q},
\end{equation}
or equivalently $|\beta_{y}\rangle=\sum_{x'} F_{yx'}|\alpha_{x'}\rangle$, 
with the same matrix $F$ as in the proof of Theorem \ref{Our Theorem 1}.
This shows that Bob's optimal quantum strategy (as encoded in the vectors $\{\ket{\beta_y}\}$)
is fully determined by Alice optimal strategy (the vectors $\{\ket{\alpha_x}\}$). 

In order to bound the dimensionality of the behaviours that determine the corresponding 
face we note that 
$c_{xy} = \braket{\alpha_{x}}{\beta_{y}}=\sum_{x'} F_{yx'}\braket{\alpha_{x}}{\alpha_{x'}}$.
The dimension of the \aw{affine} span of the optimal behaviours $C$ is clearly bounded 
by that of the Gram matrices $[\braket{\alpha_{x}}{\alpha_{x'}}]_{xx'}$, which we
recall to be real symmetric matrices with diagonal elements
$\braket{\alpha_{x}}{\alpha_{x}}=1$, thus the dimension is bounded 
by $\frac12 m_{A}(m_{A}-1)$ and therefore the codimension of the face 
of $\cQ_0$ is bounded by the classical bound derived in Theorem \ref{Our Theorem 1}. 
To get the same for $\cQ_{\text{com}}$ and $\cQ_\otimes$, we need to 
control also the marginal parts of the behaviours, which can be written
as $\alpha_x = \braket{\psi}{\alpha_x}$ and 
$\beta_y = \braket{\psi}{\beta_y} = \sum_{x'} F_{yx'} \alpha_{x'}$.
In other words, $\ket{\beta} = F\ket{\alpha}$ is fully determined by 
$\ket{\alpha}$, just as we had seen for the deterministic classical strategies,
hence only $m_A$ is added to the dimension of the face.

For non-exhaustive XOR games we then proceed as in Theorem \ref{Theorem2}.
In particular, XOR games do not define facets of either 
$\cQ_{\text{com}}$ or $\cQ_\otimes$. 
In the quantum correlation set $\cQ_0$, the XOR games that can define facets
have only one input each for Alice and Bob, and that leaves only $|c_{xy}|\leq 1$,
which are indeed facet defining inequalities, but they are trivial, as they 
correspond to the non-negativity of probability.

To arrive at the conclusion for $\cQ_0$ having no facets at all, note that 
any purported facet is exposed, so it has to be defined by an XOR game,
and for those we just showed that only the trivial inequalities define facets.
\end{proofof}

\subsection{D. Dimension bound for non-local computation and (asymptotic) attainability}
In \cite{linden_quantum_2007}, Linden \emph{et al.} introduced the cooperative games of non-local computation (NLC) and showed that quantum strategies provide no advantage over classical ones, although stronger forms of non-signaling correlations allowed perfect success. In the problem of non-local computation Alice and Bob need to collaborate in order to compute to  a boolean function $f(z)$ of a string of $n$ bits,  $z=z_1z_2\ldots z_n$,  without communicating with each other during the computation, and without individually learning anything about the input string $z$. The inputs are promised to be given with an arbitrary probability distribution $\tilde{q}(z)$ known to Alice and Bob and are split in two correlated signals $x=x_1x_2\ldots x_n$ and $y=y_1y_2\ldots y_n$ so that $z=x\oplus y$ (meaning $x_i\oplus y_i=z_i$ for all $i=1,\ldots,n$)  which are given to Alice and Bob respectively.  In order to enforce that  Alice and Bob do not learn anything about $z$, it is necessary that  $p(x_i=0)=p(x_i=1)=\frac12$ for all $x_i$ and idem for $y_i$.   It is hence clear that NLC is a particular instance of an XOR game with $q(x,y) = \frac{1}{2^n}\tilde{q}(x\oplus y)$ and $a\oplus b= f(z=x\oplus y)$.The corresponding game matrix is given by $\Phi_{NLC}=\sum_{x,y}(-1)^{f(x\oplus y)}\tilde{q}(x\oplus y)|x\rangle\!\langle y|$.

In \cite{linden_quantum_2007}, 
it is proved that the game matrix is diagonal in the Fourier (Hadamard) basis, 
$|\tilde{u}\rangle=\sum_{x}(-1)^{u \cdot x}|x\rangle$, where $u\cdot x$ is the 
inner product modulo $2$ of the bit strings $u$ and $x$. As a consequence, 
$\xi_Q \leq  2^{n-1}\|\Phi_{NLC}\| := \xi^*$, where $\|\cdot\|$ 
denotes the operator norm. 
Moreover, it was shown that the upper bound is attained by a classical local 
strategy with $\xi_{c}=\xi^*$, i.e. by a suitably chosen pair of vectors 
$\ket{\alpha_{c}}$, $\ket{\beta_{c}}$,
and one concludes that NLC games present no quantum advantage. 
Since the inequality
\begin{equation}
  \label{w leq w*}
  \xi \leq \xi^*
\end{equation} 
holds for any classical local and quantum average success probability, this Bell 
inequality is neither violated by classical physics nor quantum mechanics and 
thus we have that the Bell polytope and the quantum set share a region of their 
boundary.  

The result stated below as Corollary \ref{Theorem NLC no facet} 
(because it follows from Theorem \ref{Our Theorem 1}) was first proved
in \cite{ramanathan_tightness_2017}. It says that an NLC 
Bell inequality is never tight for any number of inputs. 
Here we present an alternative proof, based on directly bounding the dimension 
of the face defined by the Bell inequality. 
This is a reconstruction of the earlier unpublished proof referenced 
in \cite{winter_quantum_2010}.

\begin{corollary}
  \label{Theorem NLC no facet}
  For \aw{any number $n$ of input bits}, an NLC game never defines a facet of 
  the Bell polytope, nor of the correlation polytope.  
\end{corollary}

\begin{proof}
\aw{As $m_A=m_B=2^n$,} the dimension of the Bell polytope is $D=4^n+2^{n+1}$. 
In order to see that the NLC Bell inequality does not define a facet, we 
show that the dimension of the affine space generated by the optimal 
classical strategies is strictly smaller than $D-1$. 
The local classical success probability is bounded by $\omega^*=\frac12(1+\xi^{*})$, which 
only depends on $\lambda = \|\Phi_{NLC}\|$.
Let $|u_0\rangle$ be a corresponding eigenvector of $\pm\lambda$.

Now, $\xi$ achieves the maximum value when $|\alpha_c\rangle$ and $|\beta_c\rangle$ 
are both proportional to $|u_0\rangle$, i.e. $|\alpha_c\rangle=\pm|\beta_c\rangle=\pm|u_0\rangle$.
Thus, we consider the eigenspaces of the two eigenvalues $\pm\lambda$, which 
we denote by $\text{Eig}(\lambda)$ and $\text{Eig}(-\lambda)$, respectively. 
There are two cases: either $|\alpha_c\rangle = |\beta_c\rangle\in \text{Eig}(\lambda)$ 
or $|\alpha_c\rangle = -|\beta_c\rangle\in \text{Eig}(-\lambda)$. These we can 
write in a single equation as $\ket{\beta_c} = F\ket{\alpha_c}$, with 
$F = 2\Pi_{\lambda}-\1$ and $\Pi_{\lambda}$ the projector onto the
eigenspace $\text{Eig}(\lambda)$.

Therefore, the face $\mathcal{F}_{NLC}$ that these optimal strategies 
define, is contained in the following affine subspace:
\begin{equation}\begin{split}
  \aff &\left\{(\ket{\alpha_c}\!,F\ket{\alpha_c}\!,\ketbra{\alpha_c}{\alpha_c}F) 
                                               : \ket{\alpha_c} \in \{\pm 1\}^{2^n} \right\},
\end{split}\end{equation}
{whose dimension we have already bounded before, in the proof 
of Theorem \ref{Our Theorem 1}, Eq. \eqref{eq:span1}, and so
\begin{equation}
  \label{NLC smaller dimension}
  \dim \mathcal{F}_{NLC} \leq 2^n + {2^{n-1}(2^n-1)},
\end{equation}
and codimension $\Delta \geq 2^n + {2^{n-1}(2^n+1)} > 1$. 

For the correlation polytope, we get similarly that the dimension of the
face is upper bounded by ${2^{n-1}(2^n-1)}$, resulting in a 
bound of $\Delta_C \geq {2^{n-1}(2^n+1)} > 1$ for the codimension
of the face.}

{The bound Eq. \eqref{NLC smaller dimension} is achievable with
equality for $\Phi = 2^{-n}\1$; this describes a game where $x$ 
is uniformly distributed, and $y=x$, and to win, Alice and Bob have to 
output the same bit $a=b$. This is evidently possible with probability 
$1$, using any local strategy $\ket{\alpha_c}$ for Alice and $\ket{\beta_c}=\ket{\alpha_c}$
for Bob, i.e. $\alpha_x=\beta_x$ for all $x\in\{0,1\}^n$, so that $F=\1$ above.
Thus, the resulting face attains Eq. \eqref{NLC smaller dimension} with
equality; likewise, the corresponding face in the correlation polytope has 
dimension $\dim\cF_C = {2^{n-1}(2^n-1)}$. 
Of course, one might object that this game has winning probability $\omega_c=1$,
equal to the no-signalling bound, so in some sense it is trivial, but it is 
worth noting that the Bell inequality is not a trivial one ($|c_{xy}|\leq 1$).}
\end{proof}

\medskip 
{We can get a slightly better bound, sometimes much better depending on the game 
matrix $\Phi$, by exploiting the fact that the latter is Hermitian and that 
the optimal local strategies must lie either in $\text{Eig}(\lambda)$ or 
in $\text{Eig}(-\lambda)$; denote their dimensions by $k$ and $\ell$, respectively, 
so that $k+\ell \leq 2^n$.
In the following we can discard the extreme cases $k=2^n$ and $\ell=2^n$,
since those correspond to $\Phi \propto \1$, which we have just discussed. 

From the previous analysis, we have 
\begin{align}
  \mathcal{F}_{NLC} 
       &= \conv \Bigl\{(\ket{\alpha_c}\!,\epsilon\ket{\alpha_c}\!,\epsilon\ketbra{\alpha_c}{\alpha_c}) 
                      : \epsilon=\pm 1, \Bigr. \nonumber \\
        \label{eq:face}
       &\phantom{=======:}\Bigl.
                        \ket{\alpha_c} \in \text{Eig}(\epsilon\lambda) \cap \{\pm 1\}^{2^n} \Bigr\} \\
       &\subset \spn \Bigl\{(\ket{\alpha}\!,\epsilon\ket{\alpha}\!,\epsilon\ketbra{\alpha}{\alpha}) 
                      : \epsilon=\pm 1, \Bigr. \nonumber \\
        \label{eq:span}
       &\phantom{=============:}\Bigl.
                        \ket{\alpha} \in \text{Eig}(\epsilon\lambda) \Bigr\} \\
       &\!\!\!\!\!\!\!\!\!\!\!\!\!\!
        = \spn \Bigl\{(\ket{\alpha}\!,-\ket{\alpha}\!,-\ketbra{\alpha}{\alpha}) 
                                      : \ket{\alpha} \in \text{Eig}(-\lambda) \Bigr\} \nonumber \\
        \label{eq:oplus-span}
       &\!\!\!\!\!\!\!
        \oplus 
          \spn \Bigl\{(\ket{\alpha}\!,\ket{\alpha}\!,\ketbra{\alpha}{\alpha}) 
                                      : \ket{\alpha} \in \text{Eig}(\lambda) \Bigr\}.
\end{align}
Note that the linear span, in Eq. \eqref{eq:span}, denote it $\cL$, has a
dimension at least $1$ larger than the face $\mathcal{F}_{NLC}$, because the
affine span of the latter does not contain the origin. This is due to the 
fact that otherwise the optimal classical winning probability were $\frac12$, 
corresponding to a bias $0$, but it is easily seen that XOR game always have 
some positive bias \cite{PR-limits}.

Now, the spaces in Eq. \eqref{eq:oplus-span} have dimension $\ell+\frac12 \ell(\ell+1)$
and $k+\frac12 k(k+1)$, respectively, and so 
\begin{equation}\begin{split}
  \label{eq:k-l-bound}
  \dim\cF_{NLC} &\leq \dim \cL - 1 \\
                &=    k+\ell + \frac{k(k+1)}{2}+\frac{\ell(\ell+1)}{2} - 1.
\end{split}\end{equation}
Among the pairs with $k+\ell \leq 2^n$, and -- as explained before -- excluding 
$\Phi\propto\1$, i.e. imposing $k,\ell < 2^n$, the r.h.s. is maximised 
at $k=2^n-1$, $\ell=1$, for which values it reproduces the previously obtained 
bound $2^n + 2^{n-1}(2^n-1)$. Note that this restricts the game severely,
since the game matrix has only the two possible eigenvalues $\pm\lambda$, 
w.l.o.g. with multiplicities $2^n-1$ and $1$, respectively:
this means $\Phi = \lambda F = \lambda(2\Pi_{\lambda}-\1)$ and 
$\1-\Pi_{\lambda}$ has rank $1$. 
In all other cases, Eq. \eqref{eq:k-l-bound} is strictly better than the 
bound \eqref{NLC smaller dimension}. 
For the correlation polytope, we get similarly that the dimension of the
face is upper bounded by 
$\frac{k(k+1)}{2}+\frac{\ell(\ell+1)}{2} - 1 \leq 2^{n-1}(2^n-1)$.

These bounds can be attained, if not exactly, then asymptotically, as
we will show on the example $k=2^n-1$, $\ell=1$. By the analysis
of \cite{winter_quantum_2010}, this essentially leaves only the game matrix
$\Phi = \lambda(\1-2^{-n}J) - \lambda2^{-n}J$, where $J=\ketbra{\bar 1}{\bar 1}$ is the all-one-matrix, and $\ket{\bar 1}=\sum_{x=1}^{2^{n}} \ket{x}.$
Since the sum of the absolute values of the entries of $\Phi$ has be to
$1$, this fixes the value of $\lambda = \frac{1}{3\cdot 2^n - 4}$,
corresponding to the winning probability (classical and quantum) 
$\omega^* = \frac12\left( 1 + \frac{2^n}{3\cdot 2^n - 4}\right) \approx \frac23 < 1$
for $n \geq 2$. I.e. the game, and with it its Bell inequality, is nontrivial.
The game can be described as follows: $x$ and $y$ are jointly distributed  
according to $q(x,y)>0$, and to win, Alice and Bob have to output the same 
bit $a=b$ if $x=y$, and different bits $a\neq b$ if $x\neq y$.
The optimal classical local strategies are on the one hand 
$\ket{\alpha_c} = \ket{\bar 1} = -\ket{\beta_c}$
(corresponding to the single negative eigenvalue $-\lambda$ of $\Phi$),
and $\ket{\alpha_c} = \ket{\beta_c} \perp  \ket{\bar{1}}$
(corresponding to the $(2^n-1)$-fold eigenvalue $\lambda$). The latter
means that $\ket{\alpha_c} = \ket{\beta_c}$ has to have exactly $2^{n-1}$
entries $+1$ and $2^{n-1}$ entries $-1$, of which there are ${2^n \choose 2^{n-1}}$. 
With this, we determine the dimension of the corresponding face of the correlation Bell 
polytope as
\[
  \dim\cF = 1 + \dim \aff \left\{ \ketbra{\alpha_c}{\alpha_c} 
                                  : \ket{\alpha_c} \perp (\ket{0}+\ket{1})^{\otimes n} \right\},
\]
where it is understood that $\ket{\alpha_c} \in \{\pm 1\}^{2^n}$. The 
affine span on the r.h.s. is precisely the space $\mathcal{G}_0$ of real 
symmetric matrices with $1$'s along the diagonal and with the property that 
all row and column sums are $0$. By parameter counting, it is straightforward
to see that $\dim\mathcal{G}_0 = 2^{n-1}(2^n-3)$, and so we get 
$\dim\cF = 1 + 2^{n-1}(2^n-3) \sim 2^{n-1}(2^n-1)$, matching the upper 
bound to leading order. 
\hfill\qed}

\end{document}